\documentclass[twoside,leqno]{article}

\usepackage[letterpaper]{geometry}

\usepackage{ltexpprt}
\usepackage{hyperref}
\usepackage{amsmath}
\usepackage{amssymb}
\usepackage{amsfonts}
\usepackage{mathtools}

\usepackage[backend=biber, sorting = none, maxnames=10, minalphanames=3]{biblatex}
\begin{document}

\newcommand\relatedversion{}
\renewcommand\relatedversion{\thanks{The arXiv version of the paper can be accessed at \protect\url{https://arxiv.org/abs/2308.12483}}} 

\title{\Large Linear-Sized Spectral Sparsifiers and the Kadison-Singer Problem\relatedversion}
\author{Phevos Paschalidis\thanks{Harvard University. \href{mailto:ppaschalidis@college.harvard.edu}{ppaschalidis@college.harvard.edu}}
\and Ashley Zhuang\thanks{Harvard University. \href{mailto:azhuang@college.harvard.edu}{azhuang@college.harvard.edu}}}

\date{}

\maketitle


\fancyfoot[R]{\scriptsize{Copyright \textcopyright\ 2024 by SIAM\\
Unauthorized reproduction of this article is prohibited}}





\begin{abstract} \small\baselineskip=9pt 
The Kadison-Singer Conjecture, as proved by Marcus, Spielman, and Srivastava (MSS) \cite{marcus2015interlacing}, has been informally thought of as a strengthening of Batson, Spielman, and Srivastava’s theorem that every undirected graph has a linear-sized spectral sparsifier \cite{batson2012twice}. We formalize this intuition by using a corollary of the MSS result to derive the existence of spectral sparsifiers with a number of edges linear in their number of vertices for all undirected, weighted graphs. The proof consists of two steps. First, following a suggestion of Srivastava \cite{srivastava2013discrepancy}, we show the result in the special case of graphs with bounded leverage scores by repeatedly applying the MSS corollary to partition the graph, while maintaining an appropriate bound on the leverage scores of each subgraph. Then, we extend to the general case by constructing a recursive algorithm that repeatedly (i) divides edges with high leverage scores into multiple parallel edges and (ii) uses the bounded leverage score case to sparsify the resulting graph.
\end{abstract}

\section{Introduction}

In the design of fast and space-efficient algorithms, one powerful concept is that of graph sparsification. Rather than run a computationally expensive algorithm on a potentially dense graph $G = (V, E)$ with $|V|=n$ and $|E| = m$, one can instead use a new, sparser graph $\widetilde{G} = (V, \widetilde{E})$ with $|\widetilde{E}| \ll m$ that approximates the original with respect to some important properties, usually obtained by selecting --- and potentially reweighting --- some of $G$'s edges. In one of the first instances of sparsification, Benczur and Karger \cite{benczur1996approximating} introduced the notion of cut-sparsifiers, and used them to improve the asymptotic runtime of state-of-the-art minimum $s$-$t$ cut and sparsest cut approximation algorithms by first cut-sparsifying the input graph and then running the algorithms.

\subsection{Spectral Sparsification.}

In the seminal paper of Spielman and Teng \cite{spielman2011spectral}, they introduced a spectral form of graph sparsification to address problems in numerical linear algebra and spectral graph theory; this notion of sparsification involves approximating the Laplacian. The Laplacian matrix of an undirected, weighted graph $G$ is $L_G=D-A$, where $D$ is the diagonal matrix of degrees and $A$ is the weighted adjacency matrix, and it encodes fundamental information about the graph's cuts, random walks, etc. Equivalently, we can define it as a sum of rank-one matrices:

\begin{Definition}[The Laplacian]\normalfont \label{def:lapl}
    The \emph{Laplacian matrix} of an undirected, weighted graph $G = (V, E, w)$ is
    \[
    L_G = \sum_{\{i, j\}\in E} w_{ij} (\delta_i - \delta_j)(\delta_i - \delta_j)^T = \sum_{\{i, j\} \in E} b_{ij} b_{ij}^T,
    \]
    where $\delta_i$ is the $i^{\textrm{th}}$ standard basis vector and $b_{ij} \coloneqq \sqrt{w_{ij}}(\delta_i - \delta_j)$ is the \emph{weighted incidence vector} for $\{i,j\}$.
\end{Definition}
The Spielman and Teng notion of \emph{spectral sparsification} requires
that the Laplacian quadratic form of the sparsifier approximates that of the original graph. Formally,



\begin{Definition}[Spectral approximation] \normalfont
\label{def:spec-approx}
Given undirected, weighted graphs $G$ and $H$, we say $H$ is an $\epsilon$\emph{-spectral approximation} of $G$ if
\[
(1-\epsilon)L_G \preceq L_H \preceq (1+\epsilon)L_G,
\]
where $\preceq$ represents the \emph{Löwner order}; i.e., $A \preceq B$ for Hermitian $A,B$ if $B-A$ is positive semidefinite.
\end{Definition}

\noindent This notion of sparsification is strictly stronger than that of cut sparsifiers ---  a spectral sparsifier automatically satisfies the requirements of Benczur and Karger's cut sparsifier definition. Despite this, Spielman and Teng were still able to show the existence of sparsifiers with $\widetilde{O}(n/\epsilon^2)$ edges for any undirected graph, and indeed used their result to design nearly-linear time algorithms for solving diagonally-dominant linear systems \cite{spielman2011spectral, spielman2006linear}. Their result was soon improved upon by Spielman and Srivastava \cite{spielman2011graph}, who used a random sampling technique to derive sparsifiers of size $O(n\log n/\epsilon^2)$ in nearly-linear time, and then again by Batson, Spielman, and Srivastava (BSS) \cite{batson2012twice} who showed a deterministic, polynomial-time algorithm for finding linear-sized sparsifiers. Specifically, they show a stronger version of the following main theorem\footnote{In fact, the number of edges in the construction BSS proves has only twice as many edges as the Ramanujan graph, a well-known, excellent sparsifier for the complete graph. Their proof also holds for larger $\epsilon$.}:
\begin{theorem}[Weaker version of Theorem 1.1, \cite{batson2012twice}] \label{th:BSS}
For every $0 < \epsilon < 1$, every undirected, weighted graph $G$ with $n$ vertices contains a reweighted subgraph $H$ with $O(n/\epsilon^2)$ edges such that $H$ is an $\epsilon$-spectral approximation of $G$.
\end{theorem}

\subsection{Connections to the Kadison-Singer Problem.} 

In the conclusion of \cite{batson2012twice}, the authors make an interesting connection between their main theorem and an outstanding open problem in mathematics: the Kadison-Singer conjecture, which dates back to 1959 \cite{kadison1959extensions}. Using a reformulation of the Kadison-Singer problem due to Weaver \cite{WEAVER2004227}, BSS conclude that stronger version of their theorem, one in which all the edges of the original graph were either discarded or reweighted identically, would imply a positive solution to the Kadison-Singer problem \cite{batson2012twice}. Though the Kadison-Singer conjecture was not proven in this way, a positive solution was shown a few years later by Markus, Spielman, and Srivastava (MSS) \cite{marcus2015interlacing} who used the Weaver equivalence and a multivariate generalization of the argument made by BSS. 

Following their breakthrough, Srivastava, author of both the BSS and MSS papers, discussed the similarity between the two results in a newsletter \cite{srivastava2013discrepancy}. Specifically, Srivastava outlined a proof for linear-sized sparsifiers of the unweighted complete graph that is based on the following implication of the MSS result. 

\begin{theorem}[Theorem 2, \cite{srivastava2013discrepancy}. Implied by Corollary 1.5, \cite{marcus2015interlacing}]
\label{th:KS}
    Given vectors $v_1, \dots, v_m \in \mathbb{R}^n$, there exists a partition $T_1 \sqcup T_2 = [m] = \{1, \ldots, m\}$, such that for $j=1,2$,
    \[
        \left(\frac{1}{2} - 5\sqrt{\alpha}\right)\left(\sum_{i=1}^m v_iv_i^T\right) \preceq \sum_{i \in T_j} v_iv_i^T \preceq \left(\frac{1}{2} + 5\sqrt{\alpha}\right)\left(\sum_{i=1}^m v_iv_i^T\right)
    \]
    where $\alpha = \max_i v_i^T(\sum_{i=1}^m v_iv_i^T)^{+}v_i$ and $A^+$ denotes the Moore-Penrose Pseudoinverse of $A$.
\end{theorem}

\noindent Intuitively, Theorem \ref{th:KS} claims that any group of vectors can be partitioned into two such that both subsets contribute approximately equally to the quadratic form of $v_1, \dots, v_m$ in any direction. Note that the approximation factor depends on $\alpha$, which measures the maximum fraction of the quadratic form that a single $v_i$ contributes. By taking the vectors $v_i$ in Theorem \ref{th:KS} to be the weighted incidence vectors $b_{ij}$ of some graph $G$ as in Definition \ref{def:lapl}, we recover a Laplacian relationship resembling Definition \ref{def:spec-approx}. Moreover, $\alpha$ is then exactly the maximum leverage score in the graph, where the leverage score $\ell_e \in [0,1]$ is a measure of the relative importance of edge $e$ in connecting its endpoints. We recall its formal definition below.

\begin{Definition}[Leverage Score]\normalfont
    For an edge $e = \{i,j\}$ in an undirected, weighted graph $G = (V, E, w)$ with Laplacian $L_G$, the \emph{leverage score} of $e$ is
    \[
    \ell_e = w_{ij} \cdot (\delta_i - \delta_j)^T L_G^+ (\delta_i - \delta_j) = b_{ij}^T \left(\sum_{\{i,j\} \in E} b_{ij} b_{ij}^T\right)^+ b_{ij},
    \]
\end{Definition}

In \cite{srivastava2013discrepancy}, Srivastava identifies this link between Theorem \ref{th:KS} and the definition of spectral sparsification, and then demonstrates how to repeatedly apply Theorem \ref{th:KS} to show Theorem \ref{th:BSS} for the special case of $G$ being the unweighted complete graph. In the proof, he exploits the fact that the leverage scores in the complete graph are very small (in fact, they are all equal to $2/n$) to maintain the approximation factor.

\subsection{Our Contribution.}
In this paper, we extend Srivastava's work by using Theorem \ref{th:KS} to derive Theorem \ref{th:BSS} in its full generality. In Section \ref{sec:n/m}, we explicitly show Srivastava's claim that his proof of linear-sized sparsifiers for the special case of the complete graph holds more generally for all undirected graphs with leverage scores bounded by $O(n/m)$ \cite{srivastava2013discrepancy}. These are graphs that do not contain any edges that are disproportionately important to the graph's structure. For an illustrative counterexample, consider the dumbbell graph, which is formed by connecting two complete graphs with a single edge; this middle edge has leverage score $1 = \omega(n/m)$ and is far more important to the graph than any of its other edges. The proof in this section, guided by \cite{srivastava2013discrepancy}, entails repeatedly applying Theorem 1.2 to partition the graph $G$ many times. This repetition is necessary since applying the result once only halves the number of edges in the worst case. The approximation factor, which --- as discussed earlier --- depends on the maximum leverage score, grows worse at each step, but we are able to maintain an appropriate bound until at least one subgraph has only a linear number of edges. Though we only aim to obtain one \emph{linear-sized} sparsifier, our method partitions the graph into many spectral approximations, each of which has edges uniformly reweighted.

In Section \ref{sec:gen}, we extend this approach for the case of any arbitrary undirected, weighted graph. The challenge with simply applying the strategy used in Section \ref{sec:n/m} is that in the general case we do not have any bounds on our leverage scores --- and thus our approximation factor. If we directly apply Theorem \ref{th:KS} to the dumbbell graph, for example, one side of the partition will be left disconnected and will thus necessarily be a poor spectral sparsifier, consistent with the fact that $\alpha=1$ and hence $1/2 - 5\sqrt{\alpha} < 0$ and $1/2 + 5\sqrt{\alpha} > 1$. We circumvent this challenge by constructing a recursive algorithm that repeatedly (i) divides edges with high leverage scores into multiple parallel edges, (ii) applies the bounded leverage score sparsification result from Section \ref{sec:n/m}, and then (iii) recombines parallel edges. As a result of this repeated division and recombination, we arrive at a final subgraph whose edges have been potentially reweighted non-uniformly, unlike the case of bounded leverage scores.

Our work serves to formalize the connection between the Kadison-Singer Conjecture as proved by MSS, and the BSS result that every graph has a linear-sized spectral sparsifier. While Srivastava \cite{srivastava2013discrepancy} had shown that the MSS theorem can be used to prove linear-sized sparsifiers for the unweighted, complete graph, our work builds nontrivially upon his proof to solidify the informal intuition that MSS is a strengthening of BSS more generally. We also hope that the argument presented here for the existence of linear-sized sparsifiers, while unable to match the ``twice-Ramanujan'' size bound proved by BSS,
is simpler to understand than the barrier function argument they presented, though of course it is based on the deep result of MSS. Moreover, our proof in Section \ref{sec:gen} does not depend on the MSS result directly, but rather demonstrates a technique to extend a sparsification result from graphs with bounded leverage scores to arbitrary weighted graphs, which may be of independent interest.



\subsection{Related Work.}


Following the MSS result, multiple surveys have been written in an attempt to enumerate the far-reaching consequences of their groundbreaking proof of the Kadison-Singer problem \cite{casazza2016consequences, harvey2013introduction, bownik2018kadison}. Along with Srivastava's discussion in \cite{srivastava2013discrepancy}, \cite{marcus2014ramanujan} also remarked on the similarity between the BSS and MSS papers, but their discussion focused mainly on the proof techniques rather than the results themselves.

\section{Linear-sized sparsifiers for bounded leverage scores}
\label{sec:n/m}

We dedicate Section \ref{sec:n/m} to proving the following theorem, which states the existence of linear-sized sparsifiers for graphs with bounded leverage scores. Our proof specifically utilizes the implication of the MSS result discussed in Theorem \ref{th:KS}.

\begin{theorem}
\label{th:ks-bss1}
    Given some $\epsilon$ such that $0 < \epsilon < 1$ and an undirected, weighted graph $G$ on $n$ vertices with $m$ edges whose leverage scores are bounded by $O(n/m)$, $G$ has a reweighted subgraph $H$ with $O(n/\epsilon^2)$ edges, such that $H$ is an $\epsilon$-spectral approximation of $G$.
\end{theorem}


As shown by Srivastava in \cite{srivastava2013discrepancy}, we begin by applying Theorem \ref{th:KS} to the weighted incidence vectors $b_{ij}$ of the given graph $G$ in order to obtain a Laplacian relationship with approximation factor equal to the maximum leverage score in the graph. We write the outcome formally below.

\begin{lemma}\label{lem:ks-impl} Given an undirected graph $G=(V,E)$ on $n$ vertices with $m$ edges whose leverage scores are bounded by $\ell$, there exists a partition $E_1 \sqcup E_2 = E$ of the edges of $G$ such that both subgraphs $H_1=(V, E_1)$ and $H_2=(V,E_2)$ satisfy
\[
\left(\frac{1}{2} - 5\sqrt{\ell}\right)L_G \preceq L_{H_j} \preceq \left(\frac{1}{2} + 5\sqrt{\ell}\right)L_{G},
\]
for $j = 1,2$.
\end{lemma}

As discussed in the introduction, this result in and of itself is not enough to show a linear-sized sparsifier, but as long as we can maintain an appropriate bound on the leverage scores of the subgraphs, we can continue to partition each subgraph recursively through repeated application of Lemma \ref{lem:ks-impl}. This allows us to obtain sufficiently small subgraphs while still ensuring that our final approximation factor is good. 

We start by showing a relationship between the maximum leverage scores of the partitioned subgraphs at each level of recursion. The statement itself was given in \cite{srivastava2013discrepancy} by Srivastava, though without the detailed proof we provide.
\begin{lemma} \label{lem:levbounds}
    Define $\ell_i$ to be the maximum leverage score among all edges in the $2^i$ subgraphs of $G$ obtained after $i$ recursive applications of Lemma \ref{lem:ks-impl}. Assuming $\ell_{i-1}$ is sufficiently small, we have
    \begin{align}
    2\left(e^{10\sqrt{\ell_{i-1}}}\right) \ell_{i-1} 
    &\geq \left(\frac{1}{2} - 5\sqrt{\ell_{i-1}}\right)^{-1}\ell_{i-1} \label{eq:leftreclevedit} \\
    &\geq \ell_i \geq  \label{eq:centerreclevedit}\\ 
    & \left(\frac{1}{2} + 5\sqrt{\ell_{i-1}}\right)^{-1}\ell_{i-1}
    \geq \frac{3}{2}\cdot \ell_{i-1} \label{eq:rightreclevedit}.
    \end{align}
\end{lemma}

\begin{proof}
    The inequalities in (\ref{eq:centerreclevedit}) can be derived almost directly from the statement in Lemma \ref{lem:ks-impl}. If $A$ and $B$ share the same nullspace, then $A \preceq B$ implies $A^{+} \succeq B^{+}$, and thus for $j=1,2$,
    \begin{align*}
    \left(\frac{1}{2} - 5\sqrt{\alpha}\right)^{-1}L^{+}_G \succeq L^{+}_{H_j} \succeq \left(\frac{1}{2} + 5\sqrt{\alpha}\right)^{-1}L^{+}_G.
    \end{align*}
    In particular, the inequality
    \begin{align*}
    \left(\frac{1}{2} - 5\sqrt{\alpha}\right)^{-1}x^TL^{+}_Gx \geq x^TL^{+}_{H_j}x \geq \left(\frac{1}{2} + 5\sqrt{\alpha}\right)^{-1}x^TL^{+}_Gx
    \end{align*}
    holds for all $x \in \mathbb{R}^n$, and thus also holds for the $b_{ij}$'s. Applying this result to the graphs obtained after $i$ levels of partitioning, we have
    \[
    \left(\frac{1}{2} - 5\sqrt{\ell_{i-1}}\right)^{-1}\ell_{i-1} \geq \ell_i \geq \left(\frac{1}{2} + 5\sqrt{\ell_{i-1}}\right)^{-1}\ell_{i-1}
    \] as desired. Moving on to the inequality in (\ref{eq:leftreclevedit}), we have by Taylor series expansion that $\frac{1}{1-x} \leq 1 + x \leq e^x$ and therefore 
    \[ \left(\frac{1}{2} - 5\sqrt{\ell_{i-1}}\right)^{-1} = 2(1 - 10\sqrt{\ell_{i-1}})^{-1}\leq 2e^{10\sqrt{\ell_{i-1}}}. \] Finally, to show (\ref{eq:rightreclevedit}), since $\sqrt{\ell_{i-1}}$ is sufficiently small (in particular $5\sqrt{\ell_{i-1}} \leq 1/6$), we can bound \[\left(\frac{1}{2} + 5\sqrt{\ell_{i-1}}\right)^{-1}\ell_{i-1} \geq \frac{3}{2}\ell_{i-1}.\]
\end{proof}

\noindent A useful corollary follows immediately from Lemma \ref{lem:levbounds}:
\begin{corollary} \label{cor:geom}
    The following holds for all integers $k \geq 1$:
    \[
    \sum_{i=0}^{k} \sqrt{\ell_i} \leq (3 + \sqrt{6})\sqrt{\ell_k}.
    \]
\end{corollary}
\begin{proof}
    We can use (\ref{eq:rightreclevedit}) to bound the sum of the leverage scores with an infinite geometric series. That is,
    \[
        \sum_{i=0}^k \sqrt{\ell_i} \leq \sum_{i=0}^k \sqrt{(2/3)^{k-i}\ell_{k}} \leq \left(\frac{1}{1 - \sqrt{2/3}}\right)\sqrt{\ell_{k}} = (3 + \sqrt{6})\sqrt{\ell_k}
    \]
    where the last inequality follows from the convergence of the infinite geometric series with ratio $\sqrt{2/3} < 1$.
\end{proof}

\begin{Remark}
Before we continue the proof of Theorem \ref{th:ks-bss1}, we make a few remarks about Lemma \ref{lem:levbounds}. Having an upper bound on the growth of the leverage scores after each recursive partitioning step is necessary in bounding the final approximation factor since at each step $i$, our single-step approximation factor is a function of $\ell_i$. We could have used a similar technique as we did in showing inequality (\ref{eq:rightreclevedit}) to show a looser upper bound of, say, $3\ell_{i-1}$, but having a multiplicative factor of $2$ in (\ref{eq:leftreclevedit}) is essential to ensuring that our allowed recursive depth is enough to create linear-sized sparsifiers. In order to achieve the factor of 2, we need an additional factor of $e^{10\sqrt{\ell_{i-1}}}$ as well. These ultimately accumulate in the final approximation bound, incentivizing the lower bound in (\ref{eq:rightreclevedit}) that we use to bound the sum of the leverage scores in Corollary \ref{cor:geom}.
\end{Remark}

Importantly, Lemma \ref{lem:levbounds} holds only for sufficiently small $\ell_{i-1}$. Given a tight bound $\gamma$ and a looser one $\delta$ on the maximum leverage scores of the original graph, we show that a recursive depth of $t = \log(1/\gamma) - \log(2/\delta)$ will maintain the looser bound of $\delta$ for each subsequent maximum leverage score as well. In order to prove Theorem \ref{th:ks-bss1}, we will ultimately apply this bound with $\gamma = \Theta(n/m)$ and $\delta = \Theta(\epsilon^2)$, but we state it more generally here so that it can be used in Section \ref{sec:gen}.

\begin{lemma} \label{lem:recappedit}
    Let $\delta$ and $\gamma$ both be at most $(\ln 2 / 10c)^2$ and assume $\ell_0 \leq \gamma \leq \delta$. Then, we can partition $G$ via Lemma \ref{lem:ks-impl} recursively for $t = \log(1/\gamma) - \log (2/\delta)$ steps while maintaining $\ell_i \leq \delta$ for all $i \in [t]$.
\end{lemma}
\begin{proof}
    We prove the statement by induction. The base case, that $\ell_0 \leq \delta$, is an assumption of the Lemma statement. For our inductive step, assume that $\ell_{k-1} \leq \delta$. Then, using Lemma \ref{lem:levbounds} --- specifically inequality (\ref{eq:leftreclevedit}) --- and Corollary \ref{cor:geom}, we obtain
    \begin{align}
    \ell_k &\leq 2^k  \exp\left(10\sum_{i=0}^{k-1} \sqrt{\ell_i}\right) \ell_0 \nonumber
    \\ &\leq \exp\left(10c\sqrt{\delta}\right)\cdot 2^k \cdot\ell_0, \label{eq:indstepedit}
    \end{align}
    where $c = 3 + \sqrt{6}$. Given our assumption that $k \leq \log 1/\gamma - \log 2/\delta$, we can simplify (\ref{eq:indstepedit}):
    \[
    \ell_k \leq e^{10c\sqrt{\delta}} \cdot \frac{1}{\gamma} \cdot \ell_0 \cdot \frac{\delta}{2} \leq \frac{e^{10c\sqrt{\delta}}}{2} \cdot \delta \leq \delta, 
    \]
    where the second inequality follows since $\ell_0 \leq \gamma$ and the last by our assumption that $\delta \leq (\ln 2/10c)^2$.
\end{proof}

The proof for Theorem \ref{th:ks-bss1} follows quickly from the results already shown. Note that for the remainder of the paper, the constant $c$ will refer to $c = 3+\sqrt{6}$ as defined in the previous proof. 

\begin{proof}[Proof of Theorem \ref{th:ks-bss1}]
Recall that we are given an $\epsilon$ such that $0 < \epsilon < 1$ and the maximum leverage score of the original graph $G$ is $O(n/m) = \rho n/m$ for some constant $\rho$. Note that $m = \omega(n/\epsilon^2)$ (since if $m = O(n/\epsilon^2)$ we would already have an appropriately sized graph), so assuming large $n$, we can bound $\ell_0 \leq \rho n/m \leq \delta$ for $\delta = \min\{(\epsilon/10c)^2, (\ln 2/10c)^2\}$. Note that since $\epsilon < 1$ we have $\delta = \Theta(\epsilon^2).$

Thus, we can apply Lemma \ref{lem:recappedit} to partition the graph $t = \log(m/\rho n) - \log(2/\delta)$ times with $\ell_t \leq \delta$. This gives $(m/\rho n) \cdot \delta/2$ subgraphs, each a result of repeated approximations of the original graph $G$. Since the edge sets of the subgraphs are disjoint and the total number of edges sums to $m$, the smallest subgraph must be of size at most $\rho n \cdot 2/\delta = O(n/\epsilon^2)$. Let this graph be $H$. We note now that, up to a constant factor, the statements
\begin{equation*}
    (1- \epsilon)L_G \preceq L_H \preceq (1+\epsilon) 
    \quad\text{and}\quad 
    e^{-\epsilon}L_G \preceq L_H \preceq e^{\epsilon}L_G
\end{equation*}
are equivalent. Therefore, ignoring the additional constant factor, the Löwner inequality from Lemma \ref{lem:ks-impl} implies that
\begin{align*}
    2^t \cdot L_H &\preceq \exp\left( \sum_{i=0}^t 10 \sqrt{\ell_i}\right)L_G \preceq \exp\left(10c\sqrt{\delta}\right)L_G \preceq e^{\epsilon}L_G,
\end{align*}
where $\sum_{i=0}^t \sqrt{\ell_i} \leq  c\sqrt{\ell_t} \leq c\sqrt{\delta}$ by Corollary \ref{cor:geom} and Lemma \ref{lem:recappedit}. An analogous argument gives us that
\[
 2^t \cdot L_H \succeq \exp\left(-10 \sum_{i=0}^t \sqrt{\ell_i}\right)L_G \succeq e^{-\epsilon}L_G,
\]
and hence
\[
e^{-\epsilon}L_G \preceq 2^t \cdot L_H \preceq e^{\epsilon}L_G.
\]
Reweighting the edges of $H$ by a factor of $2^t = O(m\epsilon^2/n)$ thus gives a reweighted subgraph $H'$ of size $O(n/\epsilon^2)$ that is an $\epsilon$-approximation of $G$.
\end{proof}

\section{Extension to general case} \label{sec:gen}

In this section, we extend our proof of linear-sized spectral sparsifiers to the general case in which we have no bound on the leverage scores. That is, we prove the following:
\begin{theorem} [Restatement of Theorem \ref{th:BSS}] \label{th:ks-bssgen}
Given an undirected, weighted graph $G$ on $n$ vertices and $m$ edges and an $\epsilon$ such that $0 < \epsilon < 1$, there exists a reweighted subgraph $H$ of $G$ with $O(n/\epsilon^2)$ edges such that $H$ is an $\epsilon$-spectral approximation of $G$.
\end{theorem}
\noindent We will use another recursive algorithm to obtain our sparsifier based on the following Theorem:
\begin{theorem} \label{th:rsboundcase}
    Let $G = (V, E)$ be an undirected, weighted graph on $n$ vertices and $m$ edges and define some $\hat{m}, ~\delta$ such that $\hat{m} \geq m$ and $3n/\hat{m} \leq \delta \leq (\ln 2/10c)^2$. Then, there exists a reweighted subgraph $H$ of $G$ of size at most $\hat{m}/3 + 6n/\delta$ that is a $10c\sqrt{\delta}$ spectral approximation of $G$.
\end{theorem}

The proof of Theorem \ref{th:rsboundcase} is based largely on the result of Lemma \ref{lem:recappedit} and the proof of Theorem \ref{th:ks-bss1}. Since we no longer have a bound on the leverage scores of $G$, we are going to consider separately the ``good edges'' (those that can be bounded by 
$3n/\hat{m}$)
and the ``bad edges'' (those that cannot). As expected, the application of Lemma \ref{lem:recappedit} will make the number of good edges linear. Unfortunately, we have no such guarantee for our bad edges, but since we cannot have too many of them, the sparsity of the final graph does not suffer too much. The quantity $\hat{m}$ is not important to the proof itself, but will be helpful later when we define our algorithm (see Remark \ref{rem:hats}).

\begin{proof}[Proof of Theorem \ref{th:rsboundcase}]
    

    Define $S = \{e \in E : \ell_e > 3n/\hat{m}\}$ to be the set of bad edges --- those with too large leverage scores. Then, we construct a new graph $G'$ by splitting each edge $e = \{a,b\} \in S$ into $\hat{m}/3n$ parallel edges of equal weight such that the total weight of the edges from $a$ to $b$ remains the same. Note that the maximum leverage score of any single edge in the original graph $G$ is 1, and thus each edge  in $G'$ will have leverage score bounded by $3n/\hat{m} \leq \delta$.

    Invoking Lemma \ref{lem:recappedit}, we can partition $G$ recursively for $t = \log(\hat{m}/3n) - \log(2/\delta)$ steps to obtain $\hat{m}/3n ~\cdot~ \delta/2$ subgraphs with $\ell_t \leq \delta$. Since $|E - S| \leq m \leq \hat{m}$, the subgraph with the fewest edges from the set $|E - S|$ must have at most $6n/\delta$ edges from $E - S$. Let this subgraph be $H = (V, E_H)$. In the worst case, $H$ has at least one copy of each edge from the set $S$, but we can quickly bound $|S| \leq \hat{m}/3$ by contradiction: if we had $|S| > \hat{m}/3$, then $\sum_{e \in S} \ell_e > (3n/\hat{m}) \cdot (\hat{m}/3) = n > n-1$, which is a contradiction since the leverage scores of every graph sum to $n-1$. Therefore, recombining the parallel edges that came from $S$ by adjusting weights, we have
    \[
        |E_H| \leq |S| + \frac{6n}{\delta} \leq \frac{\hat{m}}{3} + \frac{6n}{\delta}.
    \]
    Using an analogous argument as in our proof of Theorem \ref{th:ks-bss1}, if we reweight each (potentially recombined edge) by a factor of $2^t$, we obtain the Löwner order inequality:
    \[
    \exp\left(-10\sum_{i=0}^t \sqrt{\ell_i}\right) L_G \preceq 2^t \cdot L_{H} \preceq \exp\left(10\sum_{i=0}^t \sqrt{\ell_i}\right) L_G. \\
     \]
     Noting from Corollary \ref{cor:geom} that $\sum_{i=0}^t \sqrt{\ell_i} \leq c\sqrt{\ell_t} \leq c\sqrt{\delta}$, this implies
     \[
     e^{-10c\sqrt{\delta}}L_G \preceq L_{H} \preceq e^{10c\sqrt{\delta}}L_G. 
    \]
\end{proof}


\noindent The next step is to apply Theorem \ref{th:rsboundcase} algorithmically. 
\begin{algorithm} \label{alg:1}
Given an $\epsilon$ such that $0 < \epsilon < 1$ and undirected, weighted graph $G = (V, E)$ where $|V| = n$ and $|E| = m$,
    \begin{enumerate}
    \item Define $G_0 = G$ and $m_0 = \hat{m}_0 = m$. \label{en:1}
    \item For a graph $G_i = (V, E_i)$ with $m_i$ edges and our chosen $\hat{m}_i \geq m_i$, choose some $\delta_i$ that satisfies the inequality $3n/\hat{m}_i \leq \delta_i \leq (\ln 2/10c)^2$. Then, use Theorem \ref{th:rsboundcase} to obtain $G_{i+1} = (V, E_{i+1})$ of size $m_{i+1} \leq \hat{m}_i/3 + 6n/\delta_i$ that is a $10c\sqrt{\delta_i}$ spectral approximation of $G_i$. Finally, for use in the next iteration, define $\hat{m}_{i+1} = \hat{m}_i/3 + 6n/\delta_i \geq m_{i+1}$. \label{it:2}
    \item We repeat this process for at most $T+1$ steps where $T = \log_{3} (m/n) - \log_{3}(1/\epsilon^2)$, terminating early if we reach at most $\beta n/\epsilon^2$ edges at any point for some constant $\beta$ defined explicitly at the end of Lemma \ref{lem:deltas}. Call this final graph $G_{T'}$.
    \end{enumerate}
\end{algorithm}  

We now need to prove that our final graph $G_{T'}$ (where $T' \leq T+1$) satisfies the appropriate size and approximation criteria. In order to do so, we will require careful selection of the $\delta_i$'s. Since the approximation bound at the $i^{th}$ step of our recursive algorithm is $\Theta(\sqrt{\delta_i})$, the final approximation of bound for $G_{T'}$ will be $\Theta(\sum_{i=0}^{T'-1} \sqrt{\delta_i})$. Thus, just as we showed a geometric series for the $\ell_i$'s in Corollary \ref{cor:geom}, we will make a similar argument for our $\delta_i$'s. Ultimately, we will define $\delta_T$ in terms of our constant $\epsilon$ and obtain a geometric series through explicit definition of the preceding $\delta_i$'s. This means, however, that $\delta_i$ is no longer a constant, so we must carefully argue that at every step before our termination point $T'$ we maintain $\delta_i \geq 3n/\hat{m}_i$, in order to justify applying Theorem \ref{th:rsboundcase}. It is easy check that each step is $\delta_i \leq (\ln 2/10c)^2$ as well.

\begin{Remark} \label{rem:hats}
    
    It is in the following proof of Lemma \ref{lem:deltas} that the need for the $\hat{m}_i$'s becomes clear due to the inverse relationship between our approximation factor and the graph size. 
    As our graph size ($m_i$) shrinks, the approximation factor in the next step ($\delta_i$) necessarily gets worse. The introduction of the $\hat{m_i}$'s, which represent our graph size in the worst case, allows us to tie our approximation factor to a more stable quantity than the graph size itself --- which can change unpredictably from one iteration to the next. Though this does slow down the graph sparsification, the fixed evolution of the $\hat{m_i}$'s is necessary in reasoning about our explicitly chosen approximation factors.

\end{Remark}

\begin{lemma} \label{lem:deltas}
    Choose $\delta_{T} = (\epsilon/10c')^2$ for $c' = c \cdot (2 + \sqrt{2})$ and $\delta_i = \delta_T/2^{T-i}$ when running Algorithm \ref{alg:1}. Then, for $i < T'$, we indeed satisfy $\delta_i \geq 3n/\hat{m}_i$.
\end{lemma}
\begin{proof}       
    By our definitions of $\delta_i$ and $T$, we can write 
    \[
    \delta_i = \frac{\delta_T}{2^{T-i}} = 2^i \cdot \delta_T \cdot \left(\frac{n}{m_0}\right)^{1/\log_2 3} \cdot \left(\frac{1}{\epsilon^2}\right)^{1/\log_2 3} = \delta_T \cdot \left(\frac{n}{m_0/3^i}\right)^{1/\log_2 3} \cdot \left(\frac{1}{\epsilon^2}\right)^{1/\log_2 3}.
    \]
    Note that since $\hat{m}_{i} \geq \hat{m}_{i-1}/3$
    by definition, we can bound
    \[\hat{m_i} \geq \frac{\hat{m}_{0}}{3^i} = \frac{m_0}{3^i},\]
    and therefore
    \begin{equation}\label{eq:deltas}
    \delta_i = \delta_T \cdot \left(\frac{n}{m_0/3^i}\right)^{1/\log_2 3} \cdot \left(\frac{1}{\epsilon^2}\right)^{1/\log_2 3} \geq \delta_T \cdot \left(\frac{n}{\hat{m}_i}\right)^{1/\log_2 3} \cdot \left(\frac{1}{\epsilon^2}\right)^{1/\log_2 3}.
    \end{equation}
    Finally, to complete our proof, we aim to show that 
    \[
    \frac{3n}{\hat{m}_i} \leq \delta_T \cdot \left(\frac{n}{\hat{m}_i}\right)^{1/\log_2 3} \cdot \left(\frac{1}{\epsilon^2}\right)^{1/\log_2 3} \leq \delta_i.
    \]        
    Substituting our choice of $\delta_T = (\epsilon/10c')^2$, this is equivalent to showing 
    \[
        \left(\frac{n}{\hat{m}_i}\right)^{1-1/\log_2 3} \leq \left(\frac{1}{c'\sqrt{300}}\right)^2 \cdot \left(\epsilon^2\right)^{1 - 1/\log_2 3}.
    \]
    We note now that since $i < T'$, we have $\hat{m}_i \geq m_i > \beta n/\epsilon^2$ (otherwise our algorithm would have already terminated), and thus
    \[
        \left(\frac{n}{\hat{m}_i}\right)^{1-1/\log_2 3} < \left(\frac{\epsilon^2}{\beta}\right)^{1-1/\log_2 3} \leq  \left(\frac{1}{c'\sqrt{300}}\right)^2 \cdot \left(\epsilon^2\right)^{1 - 1/\log_2 3},
    \]
    where the last inequality holds if we define our constant $\beta$ such that $\beta^{1 - 1/\log_2 3} > 300(c')^2$.
\end{proof}

Now that we have confirmed that our choices of $\delta_i = \delta_T/2^{T-i}$ and $\delta_T = (\epsilon / 10c')^2$ are valid, we can use these particular $\delta_i$'s in Algorithm \ref{alg:1}. We will now show that the number of edges in $G_{T'}$ is linear, which will follow mainly from the recursive definition $\hat{m}_{i+1} = \hat{m}_i/3 + 6n/\delta$.

\begin{lemma}\label{lem:linear-general} 
    Suppose we choose $\delta_i$ for $i \in [T]$ as in Lemma \ref{lem:deltas}. Then, the final graph $G_{T'}$ returned by Algorithm \ref{alg:1} has $m_{T'} = O(n/\epsilon^2)$ edges.
\end{lemma}
\begin{proof}
    In the case that $T' < T + 1$, the algorithm has terminated early, so, by definition of the algorithm, we must have $m_{T'}\leq \beta n /\epsilon^2 = O(n/\epsilon^2)$ edges.

    Thus, suppose $T' = T+1$. Then, since 
    $\hat{m}_{i+1}=\hat{m}_i/3 + 6n/\delta_i$ from
    Step \ref{it:2} of Algorithm \ref{alg:1} and $T = \log_{3} (m/n) - \log_{3}(1/\epsilon^2)$ with $\hat{m}_0 =  m$, we can unroll the recurrence to see that
    \[
    \hat{m}_{T+1} =
    \frac{\hat{m}_0}{3^{T+1}} + \sum_{i=0}^T \left(\frac{1}{3}\right)^{T-i} \frac{6n}{\delta_i} = \frac{n}{3\epsilon^2} + \sum_{i=0}^T \left(\frac{1}{3}\right)^{T-i} \frac{6n}{\delta_i}.
    \]
    Then, recalling that $\delta_i = \delta_T/2^{T-i}$ and $\delta_T = (\epsilon / 10c')^2$, this becomes
    \[
    \hat{m}_{T+1}=
    \frac{n}{3\epsilon^2} + \sum_{i=0}^T \left(\frac{2}{3}\right)^{T-i} \frac{6n}{\delta_T} = O\left(\frac{n}{\epsilon^2}\right) + O\left(\frac{n}{\delta_T}\right) = O\left(\frac{n}{\epsilon^2}\right),
    \]
    by geometric series argument. Since $m_{T'} = m_{T+1} \leq \hat{m}_{T+1}$, we are done.
\end{proof}

We must also confirm that $G_{T'}$ is a good spectral approximation. As discussed before, this proof is based on a geometric series argument guaranteed by our careful definition of the $\delta_i$'s.
\begin{lemma} \label{lem:e-approx}
    Suppose we choose $\delta_i$ for $i \in [T]$ as in Lemma \ref{lem:deltas}. Then, the final graph $G_{T'}$ returned by Algorithm \ref{alg:1} is an $\epsilon$-spectral approximation of $G$.
\end{lemma}
\begin{proof}
    By Step \ref{it:2} of Algorithm \ref{alg:1}, we know $G_{i+1}$ is a $10c\sqrt{\delta_i}$ approximation of $G_i$ for all $i$. Thus, our approximation factor of $G_{T'}$ is
    \[
    \sum_{i=0}^{T'-1}  10c \sqrt{\delta_i} \leq 10c \sum_{i=0}^{T} \sqrt{\delta_i} = 10c \sum_{i=0}^{T} \left(\frac{1}{\sqrt{2}}\right)^{T-i} \cdot \sqrt{\delta_T},
    \]
    by our definition $\delta_i = \delta_T/2^{T-i}$. By geometric series argument, this is bounded above by $10c (2+\sqrt{2}) \sqrt{ \delta_T} = \epsilon$ so we are done.
\end{proof}
Finally, Theorem \ref{th:ks-bssgen} follows immediately from the results we have shown.
\begin{proof}[Proof of Theorem \ref{th:ks-bssgen}]
    By Lemmas \ref{lem:linear-general} and \ref{lem:e-approx}, the graph returned by Algorithm \ref{alg:1} where we choose $\delta_T = (\epsilon / 10c')^2$ and $\delta_i = \delta_T/2^{T-i}$ is an $\epsilon$-approximation of $G$ of size $O(n/\epsilon^2)$.
\end{proof}

\section{Conclusions \& Further Directions}

In this paper, we explicitly formalize the longstanding assumption that the MSS proof of the Kadison-Singer conjecture \cite{marcus2015interlacing} implies the BSS result of linear-sized spectral sparsifiers \cite{batson2012twice}. Though our proof does not match the ``twice-Ramanujan'' bound in the BSS paper, this does not necessarily mean that the complete BSS result does not follow from the MSS theorem. Our main goal was to show $O(n/\epsilon^2)$ sparsifiers, so there may exist areas in our analysis where it is possible to further tighten the constants chosen for our bounds. 

Moreover, our results only depend on an implication of Corollary 1.5 from \cite{marcus2015interlacing}. Given this, a possible area for further research is whether the MSS result in its full generality can be utilized to prove linear-sized sparsifiers for other, stronger notions of spectral approximation such as the directed generalization of spectal approximation offered in \cite{cohen2017almost}, unit-circle approximation as defined in \cite{ahmadinejad2020high}, or even for singular value approximation introduced by \cite{ahmadinejad2023singular}. So far, the best known sparsifiers are of size $\widetilde{O}(n)$ for each of these definitions \cite{cohen2017almost, ahmadinejad2020high, ahmadinejad2023singular}.

\section{Acknowledgements}
Thank you to Salil Vadhan for introducing us to spectral sparsification and for the many important suggestions, comments, and edits he made while advising us throughout the course of this project. We would also like to thank Aaron Sidford and Nikhil Srivastava for their helpful advice and encouragement.

\printbibliography

\end{document}